\def\year{2020}\relax
\documentclass[letterpaper]{article} %DO NOT CHANGE THIS
\usepackage{aaai20}  %Required
\usepackage{times}  %Required
\usepackage{helvet}  %Required
\usepackage{courier}  %Required
\usepackage{url}  %Required
\usepackage{graphicx}  %Required
\usepackage{xcolor}
\usepackage{soul}
\usepackage[utf8]{inputenc}
\usepackage{amsmath,amssymb,amsthm}
\usepackage{stmaryrd}
\usepackage{mathtools}
\usepackage{xspace}
\usepackage{url}
\usepackage{multirow}
\usepackage{tikz}
\usetikzlibrary{arrows,positioning,shapes.misc,shapes.geometric}
\usepackage{blindtext}

\newtheorem{theorem}{Theorem}
\newtheorem{corollary}[theorem]{Corollary}
\newtheorem{definition}{Definition}

\newtheorem{observation}[theorem]{Observation}

\newtheorem{proposition}[theorem]{Proposition}

\begin{document}
\title{Distance-based Equilibria in Normal-Form Games}
    \author{Erman Acar\\
     Vrije Universiteit Amsterdam\\
     Amsterdam, The Netherlands \\
    erman.acar@vu.nl
\And
   Reshef Meir  \\
     Technion - Israel Institute of Technology\\
      Haifa, Israel \\
   reshefm@ie.technion.ac.il\\
    }
\maketitle
\begin{abstract}
We propose a simple uncertainty modification for the agent model in normal-form games; at any given strategy profile, the agent can access only a set of ``possible profiles'' that are within a certain distance from  the actual action profile.  We investigate the various instantiations in which the agent chooses her strategy using well-known rationales e.g., considering the worst case, or trying to minimize the regret, to cope with such uncertainty.  Any such modification in the behavioral model naturally induces a corresponding notion of equilibrium; a \emph{distance-based equilibrium}. We characterize the relationships between the various equilibria, and also their connections to well-known existing solution concepts such as Trembling-hand perfection. Furthermore, we deliver existence results, and show that for some class of games, such solution concepts can actually lead to better outcomes. 
\end{abstract}

\section{Introduction}

Decision making under uncertainty is a key issue both in game theory and in artificial intelligence. Whereas models of strict uncertainty, or absence of unique priors are common at the outset of those fields \cite{gilboa1989maxmin,dow1994nash,halpern2017reasoning,gilboa2008probability,potyka2016group},
 probabilistic and Bayesian models still have become dominant, albeit due to different reasons. In AI, the use of probabilities often leads to better performance in a wide variety of tasks (e.g. Bayesian networks for debugging and information retrieval~\cite{heckerman1995real}, Monte-Carlo methods for robot localization~\cite{thrun2001robust}, and many more). In game theory, probabilities are used first and foremost because they allow for clean modeling, and in particular the use of von Neumann-Morgenstern utilities with all the rich theory that they support. 

Several solution concepts suggested in the behavioral game theory literature tackled the problem of imperfect rationality (e.g., Cognitive hierarchy~\cite{camerer2004cognitive}, Quantal response~\cite{mckelvey1995quantal}, Trembling-hand perfect equilibria~\cite{selten75} and others), yet many of these still assume that agents optimize or approximate their expected utility over some distribution, which is not very cognitively plausible. For example, following a Trembling-hand perfect strategy, which appears to be the foremost central notion of  equilibrium refinements,  can be claimed to be a ``super-rational'' behavior rather than bounded rational~\cite{aumann1997rationality}. Recently, Trembling-hand perfect equilibrium regained strong attention in machine learning research as well  \cite{farina2018trembling,farina2018practical}. 

In this work, we suggest a model for an agent inspired by early work in AI on uncertainty and reasoning (see the work by Halpern \cite{halpern2017reasoning}), and some recent works on specific games (such as voting~\cite{conitzer2011dominating,MLR14,LMOPAAAI19} and  routing~\cite{MP15}), and look at its very foundations; normal-form games, from the lenses of this model. Our model is \emph{distance-based}, in the sense that at any given action profile, a set of ``possible profiles'' (that are close to the actual profile) is constructed w.r.t. a metric,  without assigning them any particular single probability. Then, the agent determines her action using one of the many available rationales for decision making under strict uncertainty, e.g., considering the worst case~\cite{wald1939contributions} or trying to minimize regret~\cite{savage1951theory,hyafil2004regret}. 

The intuition of such setting is  due to  imprecision caused by limitations in observations (of the agent), and the "closeness" of signals (imposed by the environment) which causes perceptual indistinguishability that comes with it.~\footnote{For instance, see \emph{random error} \cite{cohen1998introduction}. Note that these discussions also took place in philosophy and epistemology e.g., "distant trees" argument  \cite{williamson1992inexact}.} Intuitively, to capture such notions formally, one can employ a distance-based model.  Note also that such model can be considered as a bounded rationality model, as it limits the agent's reasoning about the other agents' strategies. 

It is flexible in the sense that potentially different distance metrics and decision rules can be plugged into the model, to fit specific games or types of behavior. Once we fix our behavioral model though, it naturally induces a notion of \emph{equilibrium}, which is an action profile where no agent is inclined to change her action. 

In \emph{biased games} \cite{caragiannis2014biased}, a subclass of \emph{penalty games} \cite{deligkas2016lipschitz}), players are equipped with a non-linear utility function; this is due to an additional bias term (or \emph{penalty} occurring in the utility function. The bias term itself is a real-valued function defined on the distance (via an $L_p$ norm) between the played strategy and a \emph{base strategy} (a particular strategy e.g., represents a social norm). These particular features and the results follow, stand orthogonal to our work.

Some other related works (and the references therein) that are worth mentioning: The authors in  \cite{marinacci2000ambiguous} makes use of Choquet expected utility model based on nonadditive probabilities \cite{gilboa1989maxmin}. Their pessimistic/optimistic choices \cite{marinacci2000ambiguous} share similar intuition with our worst-case/best-case responses. The notion of local rationalizibility by K. Apt \cite{apt2007many} seems related to our \emph{local-best response}; yet there it is enough for a strategy to be a best-response to a single strategy in beliefs, whereas in our case it has to be the one optimal w.r.t the whole belief set. Aghassi and Bertsimas in their work \cite{aghassi2006robust} uses robust optimization (hence worst-case scenario) to model uncertainty in payoffs. None of those works, however, uses the distance as the basic machinery, and different technical subtleties and challenges apply.
\paragraph{Contribution and paper structure}
After giving basics and the familiar equilibrium notions in the next section, we introduce our model and explore the interlinks between different variations of it. We also explore its relation to other major refinements, among others the aforementioned notion of Trembling-hand perfection. Most of our results are not metric-specific, yet in examples and some results, due to its wide-spread use and intuition, we adopt Euclidean metric. Then we demonstrate how these solution concepts apply to several common games of interest,  and provide with existence results for our notion. To underline its benefit, we introduce a class of games such that these notions potentially guarantee better outcomes. And very much in connection with that, as our final contribution, we provide a result which gives a price of anarchy bound in terms of our notion. Conclusion and future work closes the paper.

\section{Preliminaries and Notation}\label{prem}  
We define $n$-player \emph{normal-form game} $\mathcal{G}=(N, A, u)$, where
 $N$ is a finite set of $n$ players, indexed by $i$;
$A = A_1 \times \ldots \times A_n$, where $A_i$ is a finite set of \emph{actions} (or \emph{pure strategies}) available to player $i$. 

Each vector
$a = (a_1,\ldots, a_n) \in A$ is called an action profile;
$u = (u_1,\ldots,u_n)$ where $u_i : A \to \mathbb{R}$ is a real-valued \emph{utility  function} (or \emph{payoff function}) for player $i$.
A mixed strategy $\pi_i$ for player $i$ is a probability distribution over the set of available actions $A_i$ for player $i$. Further, we denote the set of \emph{mixed-strategies} for player $i$ by $\Pi_i$, which is the set of all probability distributions over the set $A_i$ of actions for player $i$. The set of mixed-strategy profiles is simply the Cartesian product of the individual mixed-strategy sets i.e., $ \Pi =  \Pi_1 \times \ldots \times \Pi_n$. We denote a (mixed-strategy) profile by $\pi \in \Pi$. Further, for a player $i$, we denote the probability that an action $a_i$ is played under mixed strategy $\pi_i$, by $\pi_i(a_i)$. The \emph{support} of a mixed strategy $\pi_i$ for a player $i$ is the set of pure strategies $\{a_i |\pi_i (a_i) > 0\}$. For a player $i$, a mixed-strategy $\pi_i$ is \emph{totally} (or \emph{completely}) \emph{mixed} if its support subsumes $A_i$. A strategy profile $\pi$ is totally mixed if its every component is totally mixed.

For simplicity, we overload the function symbol $u_i$ to define the (expected) utility $u_i$ of a strategy  profile $\pi$ for player $i$ in a normal-form game as $u_i(\pi)=\sum_{a \in A} u_i(a) \prod_{j \in N} \pi_j(a_j)$.
\footnote{Note that large $\prod$ in this expression stands for product (instead of $\Pi$, the set of mixed profiles).} 
A (mixed) strategy $\pi_i$ is a \emph{best response} to $\pi_{-i}$ if $u_i(\pi_i,  \pi_{-i}) \geq u_i(\pi'_i,  \pi_{-i})$ for every $\pi'_i \in \Pi_i$. A (mixed) strategy profile is a \emph{Nash equilibrium} ($MN$) if for every player $i \in N$, $\pi_i$ is a best-response to $\pi_{-i}$. A pure strategy Nash equilibrium ($PN$) is a $MN$ where every player's strategy has a support of cardinality 1. A \emph{totally mixed Nash equilibrium} is denoted by $TMN$. Given a profile $\pi$, \emph{Social Welfare} $SW(\pi)= \sum_{i \in N} u_i(\pi)$. And finally, \emph{Price of Anarchy} PoA for a game is defined as the ratio of the maximum social welfare (numerator) to the minimum social welfare in an equilibrium (denominator).

 \subsection{Equilibria with Mistakes and Imprecision}\label{approximateequilibria}
We mention definitions of several well-known equilibrium concepts involved with \emph{slight mistakes} or \emph{imprecision} of agents, from  the literature.  These are \emph{Trembling-Hand Perfect Equilibrium}  ($T$) from R. Selten's seminal work ~\cite{selten75}, its stronger version, \emph{Truly  Perfect Equilibrium} ($TP$) \cite{kohlberg1981some},
and \emph{Robust equilibrium} ($R$) \cite{MP05}. These concepts are of particular importance since we shall reveal their connections to the distance-based equilibrium concepts, introduced later in the next section.  

\begin{definition}[Trembling-Hand Perfect Equilibrium~\cite{selten75}]\label{trembling-def}
Given a finite game $\mathcal{G}$, a mixed strategy profile $\pi$ is \emph{Trembling-hand perfect equilibrium} if there is a sequence $\{\pi^k\}^\infty_{k=0}$ of totally mixed strategy profiles which converges to $\pi$ such that for each agent $i \in N$, $\pi_i$ is a best response to $\pi^k _{-i}$ for all $k$. 
\end{definition}

Selten's notion of Trembling-hand perfect equilibrium is based on the notion of best response which is robust against minimal mistakes (hence the term \emph{trembling hand}) of opponents, formalized by a sequence of profiles converging to the equilibrium. As it was shown by Selten \cite{selten75}, every finite game has a $T$-equilibrium. Note that the notion does not demand a best-response to every such sequence but rather only one. We shall later show that this very notion is entangled to our notions of distance-based equilibria. So is the Truly Trembling-hand perfect equilibrium, a stronger variant, as we mention next.\footnote{Another similar variation \cite{okada1981stability} that aims to strengthen Selten's Trembling-Hand is given by Okada.}

\begin{definition}[Truly perfect equilibrium~\cite{kohlberg1981some}]\label{truly-trembling-def}
Given a finite game $\mathcal{G}$, a mixed strategy profile $\pi$ is \emph{truly perfect equilibrium} if for each sequence $\{\pi^k\}^\infty_{k=0}$ of totally mixed strategy profiles which converge to $\pi$, there is a $K$ such that $\pi_i$ is a best response to $\pi^k_{-i}$ for all $k\geq K$. 
\end{definition}

It is easy to see that this notion demands a lot by requiring each action to be a best response in every sequence of profile converging to the equilibrium. There is a cost for this demand, that is, $TP$ does not  always exist \cite{kohlberg1981some} (see also Chapter 11 of \cite{Fudenberg91}).

We also define an even stronger variation; namely, \emph{Strict-$TP$} (s-$TP$).

\begin{definition}[Strict-$TP$]\label{def:strict_TP}
A strategy profile is a \emph{Strict-$TP$} if it is defined as in Def.~\ref{truly-trembling-def} except that every $\pi_i$ is  strictly better than any other response to $\pi^k_{-i}$.
\end{definition}
We also provide a slightly stronger extension of a Selten's original Trembling-hand perfect equilibrium; namely \emph{strict-$T$} (s-$T$):
\begin{definition}[Strict-$T$]
A strategy profile is a \emph{Strict-$T$} if it is defined as in  Def.~\ref{trembling-def} except that for every $i\in N$, there is an infinite subsequence of $\{\pi^k\}_{k=0}^\infty$ where $\pi_i$ is a strict best-response to $\pi^k_{-i}$.
\end{definition}
Note that a strict-$T$ has to be a pure Nash equilibrium, since a strict best-response cannot be mixed. 

The following notion of Robust equilibrium is an adaption from robust political equilibrium \cite{MP05}.\footnote{\cite{MP05} deal with both coalitional stability and noisy actions, assuming that each player ``fails'' to play with some small probability. We only focus on the latter part. Our definition of Robust equilibrium is based on their informal description and motivation.} Intuitively, an equilibrium is $\epsilon$-Robust, if each player would like to keep her action, even if there is a small chance that other players deviate.

\begin{definition}[Robust Equilibrium~\cite{MP05}]\label{robusteq} 
 A mixed profile $\pi$ is an $\epsilon$-noisy variant of a pure profile $a$, if  for all $j\in N$, $\pi_j(a_j)>1-\epsilon$ where $\epsilon >0$.
 \begin{itemize}
    \item  Given a pure strategy profile $a$, player $i$, and $\epsilon>0$, action $b_i$ is an \emph{$\epsilon$-Robust response} if $b_i$ is a best response to any $\epsilon$-noisy variant of $a_{-i}$.
    \item A pure strategy profile $a$ is an $\epsilon$-Robust equilibrium if every $a_i$ is an $\epsilon$-Robust response to $a_{-i}$.
\end{itemize}
\end{definition}
Next, we provide a link between those two concepts.
\begin{proposition}

If $a$ is an $\epsilon$-Robust equilibrium for some $\epsilon>0$, then $a$ is a TP.
\end{proposition}
\begin{proof}
Let $a$ be some $\epsilon$-Robust equilibrium for some $\epsilon>0$, and consider a sequence $\{\pi^k\}_{k=0}^\infty$ converging to $a$. Thus $\pi_i^k(a_i)\rightarrow 1$ for all $i\in N$. In particular, there is some $K_i$ such that for all $k>K_i$, $\pi_i^k(a_i)>1-\epsilon$. Let $K:=\max_i K_i$. Then for all $k>K$, and for all $j\in N$, we have that $\pi^k_j(a_j)>1-\epsilon$, i.e. $\pi^k$ is an $\epsilon$-noisy variant of $a$, and thus $a_{-i}$ is a best response to $\pi^k_{-i}$.
\end{proof}

\section{Distance-based Equilibria}

In the following, we define the central notions of the paper. 
\subsection{Distance-based uncertainty}

For every agent $i \in N$,  let $r_i \in \mathbb{R}^+$  be the associated \emph{ignorance factor} formalizing the intuition: Greater the ignorance factor, more ignorant/cautious the agent (about the mixed strategies of other agents).  

Given a mixed strategy profile $\pi = \langle \pi_i, \pi_{-i}\rangle$, $\mathcal{B}_i(\pi, r_i):= \{\pi'_{-i} \mid d(\pi_{-i}, \pi'_{-i}) \leq r_i \}$ is the set of \emph{possible response profiles} of others that the agent $i$ is considering, equipped with a metric $d$ which is assumed to have the axioms of non-negativity i.e., $d(x, y) \geq 0$; identity of indiscernibles i.e., $d(x, y)=0 \iff x = y$; symmetry i.e., $d(x, y) = d(y, x)$; and triangle inequality i.e.,  $d(x, z) \leq d(x, y) + d(y, z)$  for all $x, y, z \in \Pi$. Note that although our results comply with any compact space with a metric which fulfills those axioms, in the examples throughout the paper we work with Euclidean metric for convenience (due to its wide-spread use).  Often, we will use the shorthand notation $\mathcal{B}_i(\pi)$ ($\mathcal{B}$ for ball) whenever $r_i$ is clear from the context.  

Intuitively, $\mathcal{B}_i(\pi)$ is the set that captures $i$'s \emph{belief} or \emph{subjective uncertainty} about other agents' strategies in a strategy profile $\pi$, hence we will call it often \emph{belief set}. Yet another description is that instead of writing a belief as a distribution over profiles, the belief of agent $i$ is written as a \emph{point estimate} $\pi_{-i}$ plus an uncertainty parameter $r_i$ (which together induce a set $\mathcal B_i$). The use of a ball to capture uncertainty is both motivated by its formal simplicity, and the degree of freedom it provides which is set aside from any obvious domain specific/context-dependent constraints.

Notice that as $r_i$ approaches 0, $i$ becomes almost sure about other agents' strategies, and thus $\mathcal B_i$ becomes $\{\pi_{-i}\}$. Noteworthy is that for two-player games $r$ reduces to the distance between two probability distributions. For more than two players, any distance on probabilities induces a natural metric $d$ on uncorrelated profiles where for each $j\neq i$ we consider all $\pi'_j$ close to $\pi_j$.

\subsection{Local responses}
Let $\Pi_i$ denote the set of all strategies available to $i$. First, we introduce some notions of best response.

\begin{definition}\label{def:LD} $\pi'_i$ \emph{locally dominates} $\pi_i$ in the set $\Pi^*_{-i} \subseteq \Pi_{-i}$ if: (a) for all $\pi_{-i}\in \Pi^*_{-i}$, $u_i( \pi'_i, \pi_{-i})\geq u_i(\pi_i, \pi_{-i})$; and (b) there exists $\pi'_{-i}\in~\Pi^*_{-i}$ such that $u_i(\pi'_i, \pi'_{-i})> u_i(\pi_i, \pi'_{-i})$.  
$\pi'_i$ \emph{strictly locally dominates}   $\pi_i$ in $\Pi^*_{-i}$ if (a) holds with strict inequality. 
\end{definition}
Note that when $\Pi^*_{-i}=\Pi_{-i}$, local dominance and strict local dominance boil down to weak and strict strategic dominance, respectively~\cite{SBmultiagent}.

Given a mixed strategy profile $\pi$, for each $i \in N$ with $r_i$, a strategy $\pi_i$ is a distance-based  
\begin{itemize}
\item[$\blacktriangleright$] \textbf{(W)orst-case} best response (or \emph{maximin})  if \\ $\pi_i = \arg\max_{\pi'_i \in \Pi_i} \{ \min(u_i(\pi'_i, \pi_{-i})) \mid \pi_{-i} \in \mathcal{B}_i(\pi)\}$.

\item[$\blacktriangleright$] \textbf{(B)est-case} best response (or \emph{maximax})  if \\$\pi_i = \arg\max_{\pi'_i \in \Pi_i} \{ \max(u_i(\pi'_i, \pi_{-i})) \mid \pi_{-i} \in \mathcal{B}_i(\pi)\}$.

\item[$\blacktriangleright$] \textbf{(WR) Worst-Case Regret} best response if \\ $\pi_i=$  $\arg\min_{\pi'_i \in \Pi_i}\max \{\texttt{reg}_i(\pi'_i, \pi_{-i})\mid \pi_{-i} \in \mathcal{B}_i(\pi)\}$ where $\texttt{reg}_i(\pi_i, \pi_{-i}) = \max_{\pi'_i \in \Pi_i}(u_i(\pi'_i, \pi_{-i})) - u_i(\pi_i,\pi_{-i})$. 

\item[$\blacktriangleright$] \textbf{(U)ndominated} best response if there is no $\pi'_i$ that  locally dominates $\pi_i$ in the set $\mathcal B_i(\pi)$. 
\item[$\blacktriangleright$] \textbf{(D)ominant} best response if $\pi_i$  locally dominates all $\pi'_i$ in the set $\mathcal B_i(\pi)$.
\item[$\blacktriangleright$] \textbf{(SD) Strictly dominant} best response if $\pi_i$  strictly locally dominates all $\pi'_i$ in the set $\mathcal B_i(\pi)$.
\end{itemize}

By the following proposition, we characterize the relations between these notions.  We make no assumption on the metric since it only uses single sets of possible strategy profiles (of opponents), a.k.a. belief sets. In other words, these relations are independent from the choice of the metric.  
\begin{theorem}\label{responsesproposition} Given any ignorance factor $r$, the following statements hold: 
\begin{enumerate}
\item[(a)]If $\pi_i$ is a $\text{SD}_r$ best response, then $\pi_i$ is a $\text{D}_r$ best response.
\item[(b)]If $\pi_i$ is a $\text{D}_r$ best response, then $\pi_i$ is a $W_r$ and $B_r$ best response. 

\item[(c)]If $\pi_i$ is a $\text{D}_r$ best response, then $\pi_i$ is a $WR_r$ best response. 

\item[(d)]If $\pi_i$ is a unique $W_r, B_r$ or $WR_r$ best response, then $\pi_i$ is a $U_r$ best response.

\end{enumerate}
\end{theorem}

\begin{proof} 
(a) Easily follows by Definition~\ref{def:LD}.

(b) Consider any action $\pi'_i\neq \pi_i$. Since $u_i(\pi_i,\pi_{-i})\geq u_i(\pi'_i,\pi_{-i})$ for any state $\pi_{-i}\in \mathcal B_i(\pi)$, this holds in particular for the states with maximal utility and minimal utility.

 (c) If $\pi_i$ is a $D_r$ best response, then $\texttt{reg}_i(\pi_i, \pi_{-i}) = \max_{\pi'_i \in \Pi_i}(u_i(\pi'_i, \pi_{-i}) - u_i(\pi_i,\pi_{-i})) \leq \max_{\pi'_i \in \Pi_i}(u_i(\pi_i, \pi_{-i}) - u_i(\pi_i,\pi_{-i}))=0$ for all $\pi_{-i}\in\mathcal B_i(\pi)$. The regret of any other action $\pi'_i$ can only be higher, thus $\pi'_i$ is a $WR_r$ response.
 
(d) Suppose that $\pi_i$ is a $B_r$ best response. If $\pi_i$ is not a $U_r$ response, then there is an action $\pi'_i\neq \pi_i$  that locally dominates $\pi_i$. In particular, $u_i(\pi'_i,\pi^*_{-i})\geq u_i(\pi_i,\pi^*_{-i})$ in the best state $\pi^*_{-i}$, which means that $\pi'_i$ is also a $B_r$ best  response. Note that uniqueness is a necessary condition, otherwise, we can consider two actions $\pi_i,\pi'_i$ that have the same utility in the best case, but one of them dominates the other. The proof for $W_r$ and $WR_r$ is similar. 
\end{proof}

The following result also holds for any metric since it only uses containment.
\begin{proposition}
If $r'_i<r_i$ then $SD_{r_i}$ response implies $SD_{r'_i}$ response.  
\end{proposition}

\begin{proof}
Since $\mathcal B(\pi,r'_i) \subseteq \mathcal B(\pi,r_i)$, condition (a) of Def.~\ref{def:LD} must hold in all states.
\end{proof}

This does not hold in any of the other variations. To see this, consider $W$-equilibrium, since for any response $\pi_i$, $\min\{u_i(\pi_i, \pi_{-i}) \mid \pi_{-i} \in \mathcal{B}_i(\pi, r)\} \leq \min\{u_i(\pi_i, \pi_{-i} \mid \mathcal{B}_i(\pi, r^*)\})\}$ whenever  $\mathcal{B}_i(\pi, r^*) \subseteq \mathcal{B}_i(\pi, r)$. The other variations are similar.

\subsection{Equilibrium}

Now, we are ready to define the notion of \textit{distance-based}  equilibrium. Let $\textbf{r} := (r_1, \ldots, r_n)$ be the \emph{ignorance vector} which stores ignorance factor for each agent $i \in N$.  Assume that $\star \in \{W, B, WR, U, D, SD\}$. Then,  $\pi$ is called a \textbf{distance-based \mbox{$\star_\textbf{r}$-equilibrium}}  if for every agent $i$, whose belief set $\mathcal{B}_i(\pi)$ is defined w.r.t. $r_i$ where $r_i = \textbf{r}_i$,  $\pi_i$ is a $\star$-best response. When all the agents have the same $r$, we will use $r$ instead of $\textbf{r}$ as a subscript, or totally omit it whenever it is clear from the context.  Similar solution concepts to $U_r,W_r,WR_r$ (in specific games) are studied in works \cite{MLR14,MP15}.

\begin{observation}
For all definitions above, if $r_i=0$ for all $i\in N$, then a $\star$-equilibrium is a Nash equilibrium. 
\end{observation}

The statements below follow  immediately from the relations between corresponding responses (i.e., Theorem~\ref{responsesproposition}).  
\begin{corollary}\label{equilibriacorollary}  Given any $r$, the following statements hold: 
\begin{itemize}
\item[(a)] If $\pi$ is a $\text{SD}_r$-equilibrium, then $\pi$ is a  $\text{D}_r$-equilibrium.

\item[(b)]  If $\pi$ is a $\text{D}_r$-equilibrium, then $\pi$ is a $W_r, B_r$ and $WR_r$-equilibrium. 

\item[(c)] If $\pi$ is a $W_r, B_r$ or $WR_r$-equilibrium, then $\pi$ is a $U_r$-equilibrium.

\end{itemize}
\end{corollary}

A brief summary of our results is illustrated in Figure~\ref{results}.
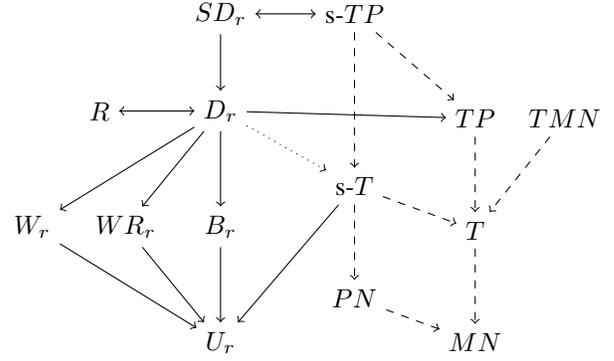
\begin{figure}

\centering

\begin{tikzpicture}[node distance=1cm]
\title{Equilibria Relations}
\node(SD)                           {$SD_r$};
\node(sTP)    [right=0.8cm of SD]                       {s-$TP$};
\node(D)       [below=0.75cm of SD] {$D_r$};
\node(TP)       [below right=0.9cm and 0.7cm of sTP] {$TP$};
\node(TMN)       [right=0.2cm of TP] {$TMN$};

\node(B)      [below=1cm of D]       {$B_r$};
\node(WR)      [left= 0.4cm of B]  {$WR_r$};
\node(W)      [left=1.8cm of B]       {$W_r$};
\node(T)      [below =1cm of TP]      {$T$};
\node(U)            [below=1cm of B]     {$U_r$};
\node(MN)            [below =1cm of T]     {$MN$};
\node(R)            [left=1 of D]     {$R$};
\node(sT)      [below=1.8cm  of sTP]       {s-$T$};
\node(PN)      [below=1cm  of sT]       {$PN$};
\draw(SD)       -- (D) [->];
\draw(D)       -- (W)[->];
\draw(D)       -- (WR)[->];
\draw(D)       -- (B)[->];
\draw(WR)       -- (U)[->];
\draw(W)       -- (U)[->];
\draw(B)       -- (U)[->];
\draw(D)       -- (TP)[->];
\draw[dashed](sTP)   -- (TP)[->];
\draw[dashed](sTP)   -- (sT)[->];
\draw[dashed](PN)   -- (MN)[->];
\draw[dashed](sT)   -- (PN)[->];
\draw[dashed](sT)   -- (T)[->];
\draw[](SD)       -- (sTP) [<->];
\draw[dotted](D)       -- (sT)[->];
\draw[](R)       -- (D)[<->];
\draw[dashed](TP)       -- (T)[->];
%\draw[dashed](D)       -- (T);
\draw[](sT)       -- (U)[->];
\draw[dashed](TMN) -- (T)[->];
\draw[dashed](T) -- (MN)[->];
\end{tikzpicture}
\caption{Entailment of equilibrium concepts; An arrow to $\star_r$ means that there is some $r>0$ for which the entailment holds. An arrow from $\star_r$ means the entailment holds for any $r>0$ (except for the dotted arrow which is slightly weaker, see Prop.~\ref{D_to_sT}). Dashed arrows mark entailments that are known or obvious.}

\label{results}
\end{figure}

\section{Locally Best Response and Trembling Hand}
Definitions~1-4 capture stability of equilibrium in a rigorous formal way, but require reasoning about sequences of profiles that do not seem have a clear cognitive interpretation. 

In this section we aim to get a better understanding of these concepts and of our distance-based equilibrium concepts, by exploring the connections between them. Moreover, the proposed distance-based best-response (and hence equilibrium) do have  cognitive interpretation which is the observational limitation that an agent has.

We first argue that robustness (under the appropriate Def.~1-4) implies stability  under uncertainty (under Def.~5) when the ignorance factors of all agents are sufficiently small.

\begin{observation}\label{obmaxeqmin} Given a profile  $\pi$ and an agent $i \in N$, for any mixed strategy  $\pi_i$, $\max_{\pi_{-i}\in \mathcal{B}_i(\pi)} u_i (\pi_i, \pi_{-i})$\\ $=\min_{\pi_{-i}\in \mathcal{B}_i(\pi)}  u_i(\pi_i, \pi_{-i}) =  u_i(\pi_i, \pi_{-i})$  if $r_i=~0$. \end{observation} 
Having Observation~\ref{obmaxeqmin} in mind, realize that the classical notion of $MN$ implies our notions of $W_r, B_r, WR_r$ and $U_r$ equilibria for $r=0$, but not for any $r>0$, as there can be (even pure) Nash equilibria that are weakly dominated. 

The following result provides a partial picture.
\begin{proposition}\label{tremblingisundominated} Given a normal-form game $\mathcal{G}$, if $\pi$ is a strict Trembling-hand perfect equilibrium, then there is an $\epsilon > 0$ such that if $r_i <\epsilon$ for every $i \in N$, then $\pi$ is a  $\text{U}_{\textbf r}$-equilibrium. 
\end{proposition}

\begin{proof}
 Assume a  game $\mathcal{G}$ and a Trembling-hand perfect equilibrium $\pi$ in $\mathcal{G}$.  By Definition~\ref{trembling-def}, there is a sequence $\{\pi^k\}^\infty _{k =0}$ of totally mixed strategies which converges to $\pi$. Take an arbitrary $\epsilon > 0$. By convergence there is a $K$ such that $d(\pi^k, \pi ) < \epsilon$ whenever $k \geq K$. Moreover, it follows that  for every $i \in N$, $d(\pi_{-i}^k, \pi_{-i} ) < \epsilon$ as well  for $k \geq K$. For each agent $i \in N$, let  $r_i =  d(\pi_{-i}^{K}, \pi_{-i} )$  i.e., $\textbf{r} = (r_1, \ldots, r_n)$. Now, for any agent $i$, we know that  $\pi_{-i}^k \in \mathcal{B}_i(\pi)$ whenever $k \geq K$, and also since $\pi$ is strict-$T$, $\pi_i$ is a best response to $\pi_{-i}^k$. 
This  shows that $\pi_i$ is $U_{r_i}$-response and that $\pi$ is a $U_{\mathbf r}$ equilibrium.
\end{proof}

In the other direction, it seems that a $D_r$-equilibrium (w.r.t. any $r>0$) must  be a strict Trembling-hand perfect equilibrium.
\begin{proposition}\label{D_to_sT}If there is an   $r^*$, such that   $\pi$ is $D_r$~-~equilibrium for all $r\in (0,r^*)$, then $\pi$ is strict-$T$.
\end{proposition}
\begin{proof}
Consider an arbitrary sequence of $r^k<r^*$ that converges to $0$.
Since $\pi$ is a $D_{r^k}$-equilibrium, then for every player $i \in N$, $\pi_i$ is a best response to every $\pi_{-i} \in  \mathcal{B}(\pi, r^k)$, and a strict best response to at least one profile $\pi^{k,i}_{-i}$. Moreover, $\pi^{k,i}_{-i}$ is w.l.o.g. totally mixed (we can mix it with a low probability for any other profile such that $\pi_i$ remains a best-response). Let $\pi^{r_k}_i$ be a mixed strategy that selects $\pi_i$ with probability $1-r_k$. 

We therefore get $n$ sequences of totally-mixed profiles converging to $\pi$, where in each sequence $((\pi^{k,i},\pi^{r_k}_i))_{k=0}^{\infty}$, $\pi_i$ is a strict best-response to the entire sequence.

Let $\pi^k = (\pi^{k,i},\pi_i)$ for all $k$ such that $k \mod n=(i-1)$ (that is, we interleave subsequences). Note that $(\pi^k)_{k=0}^\infty$ converges to $\pi$, and for every agent $i$ there is a subsequence for which $\pi_i$ is a strict best-response. Thus $\pi$ is a strict Trembling-hand perfect equilibrium (strict-$T$). 
\end{proof}

\def\calB{\mathcal{B}}
\begin{proposition}
For any $r$, if $\pi$ is $\text{D}_r$-equilibrium, then $\pi$ is TP. Similarly, $\text{SD}_r$ entails $s$-$TP$.
\end{proposition}

\begin{proof}
Fix any $r$. It follows from the fact that whatever sequence we choose in $\mathcal{B}_i(\pi_{-i}, r)$, $\pi_i$ will be a best response [respectively, strict best response] to it by definition, hence satisfying the condition of $TP$ [$s$-$TP$]. 
\end{proof}
\begin{proposition} If $\pi$ is a $TP$, then there is an $r$ such that $\pi$ is a $U_r$-equilibrium. Further, if the game also is generic then $\pi$ is a $D_r$-equilibrium.
\end{proposition}

\begin{proof}
Assume that $\pi$ is a $TP$, then for each sequence $\{\pi^k\}^\infty_{k=0}$ of totally mixed strategy profiles which converge to $\pi$, there is a $K$ such that $\pi_i$ is a best response to $\pi^k_{-i}$ for all $k\geq K$. All possible sequences form a ball for each player, and we take the infimum of $K$ of those (all) possible sequences; call it $K^\star$. Then, for all $i \in N$, we define $r_i= d(\pi_{-i}, \pi^{K^{\star}}_{-i})$. Obviously, either $\pi_i$  locally dominates all the other responses to every $ \pi'_{-i} \in \mathcal{B}_i(\pi, r_i)$ for every player $i$ (which implies that $\pi$ is a $D_r$-equilibrium), or there is some response that has the same utility as $\pi_i$ (which means non-genericity).  
\end{proof}
The following result provides a link between strict-TP and $SD_r$-equilibrium.
\begin{proposition}If  $\pi$ is a strict-$TP$ then there is an $\epsilon>0$ such that if $r_i < \epsilon$ for all $i\in N$ then $\pi$ is $SD_r$-equilibrium. 

\end{proposition}
\begin{proof}
Assume that $\pi$ is not a $SD_r$-equilibrium for any $r>0$. We will construct a sequence of states $\{\pi^k\}^\infty_{k =0}$ that converges to $\pi$, but such that for every $k$ there is some agent $i$ for which $\pi_i$ is not a strict best response to $\pi_{-i}^k$.  Let $r_k= \frac{1}{k}$. Since $\pi$ is not a $SD_r$-equilibrium, there is some $i\in N$, a profile $\pi^k_{-i}\in \mathcal{B}_i(\pi,r_k)$, and an action $a'_i$ such that $u_i(\pi^k_{-i},\pi_i) \leq  u_i(\pi^k_{-i},a'_i)$. That is, $\pi_i$ is not a strict best response to $\pi^k_{-i}$. We set $\pi^k = (\pi^k_{-i},\pi_i)$. By construction, $d(\pi^k,\pi)\leq \frac{1}{k}$ and thus $\{\pi^k\}^\infty_{k =0}$ converges to $\pi$.

\end{proof}

It seems $s$-$T$-equilibrium does not imply $D_r$-equilibrium. The proof is moved to appendix due to space restrictions.
\begin{proposition}
There is a $s$-$T$-equilibrium that is not a $B_r$ equilibrium for any $r>0$ (and thus not a $D_r$-equilibrium).
\end{proposition}

\begin{proposition}\label{robustbestcase} Given a game $\mathcal{G}$, if $\pi$ is an $R$-equilibrium, then there is an $r>0$ such that $\pi$ is a $B_r$ and $W_r$-equilibrium.
\end{proposition}

\begin{proof} We give a proof for $B_r$ ($W_r$ is similar) case.   Assume that $\pi$ is an $R$, then by definition \ref{robusteq}, there is an $\epsilon$ such that for every $\epsilon^* \leq \epsilon$, $\pi^*$ such that $\pi^* \rightarrow \pi$ is an equilibrium. Let $\pi'$ be a profile such that $d(\pi', \pi)=0$. By robustness, $\pi'$ is an equilibrium. Now, for every player $i \in N$, observe that $d(\pi', \pi ) > d(\langle \pi_i, \pi'_{-i} \rangle, \pi)$. Then, $\langle \pi_i, \pi'_{-i} \rangle \in \mathcal{B}_i(\pi', \epsilon)$. Hence, $\pi_i$ is a best response to any sequence of $\{\langle \pi_i, \pi'_{-i} \rangle \}$ converging to $\pi$. Therefore, $\pi_i$ must be a $B_r$ best response for $r= d(\langle \pi_i, \pi'_{-i} \rangle, \pi)$.\end{proof}

\begin{proposition}
Let $a$ be a pure profile in $D_r$-equilibrium, then $a$ is $\frac{r}{n^{1\slash 2}}$-Robust. Moreover, any $\epsilon$-Robust equilibrium is a $D_\epsilon$-equilibrum.
\end{proposition}

\begin{proof}
Assume that $a\in D_r$ for some $r$, and consider some $\epsilon$-noisy variant of $a$. 
$d(a_{-i},\pi_{-i})=\left(\sum_{j\neq i}(1-\pi(a_j))^2\right)^{1\slash 2} \leq (n\epsilon^2)^{1\slash 2}$, thus $\pi_{-i}\in \calB_i(a,r)$ for $r=(n\epsilon^2)^{1\slash 2}=n^{1\slash 2} \epsilon$.
This means that $a$ is $\epsilon$-robust for $\epsilon=\frac{r}{n^{1\slash 2}}$. In the other direction, suppose that $a$ is $\epsilon$-Robust and consider some $\pi$.
If $\pi(a_j)>1-\epsilon$, then $d(\pi,a)\geq ((\pi(a_j)-1)^2)^{1\slash 2} > \epsilon$, which means that all vectors in $\calB_i(a,r)$ are (at most) $r$-noisy variants of $a$. Thus for any $r\leq \epsilon$, $a_i$ is a best-response to all of $\calB_i(a,r)$, and is therefore a $D_\epsilon$ equilibrium. 
\end{proof}

The following result follows immediately from the Proposition~\ref{robustbestcase} and Corollary~\ref{equilibriacorollary}.
\begin{corollary}
Given a game $\mathcal{G}$, if $\pi$ is an $R$, then there is an $r>0$ such that $\pi$ is a $U_r$-equilibrium.
\end{corollary}

See Figure~\ref{results}, for a summary of the results obtained in this Section.

\section{Existence Results}

In the previous section, we had noted that $MN$ implies our  distance-based notions of $W_r$, $WR_r$, $B_r$ and $U_r$ equilibria as a special case when $r=0$. In this section, we deliver existence results regarding the several variants of distance-based equilibria. 

An immediate negative result to start with, $D_r$-equilibrium which is central in Figure~\ref{results}, seems to be too strong to exist in general. To see this,  assume a game in which all actions have the same payoffs for every player. Obviously any profile is a Nash equilibrium, yet none of them is a $D_r$-equilibrium (including $r=0$) due to ($b$) of Definition \ref{def:LD} of local dominance. 

\begin{corollary}
$D_r$-equilibrium does not exist in general.
\end{corollary}

Balancing out this negative news, the remaining distance-based equilibria entailed by $D_r$ do exist. 

\begin{theorem}
Every finite normal-form game $\mathcal{G}$ has a  $W_r$, $B_r$ and $WR_r$-equilibria.
\end{theorem}
We omit the actual long proof here due to space limitations, and present it in appendix. The proof idea is based on well-known application of Kakutani's fixed point's theorem i.e., defining the best response correspondence (for each one of them), and showing its convexity and upper-continuity.  It is rather straightforward in the cases of $W_r$ and $B_r$ -equilibria. In the case of $WR_r$-equilibrium, it is obtained by showing the convexity and piece-wise linearity of the worst-case best response expression (\ref{def:LD}).

Next, the following result immediately follows from the above theorem and Corollary~\ref{equilibriacorollary}.
\begin{corollary}
Every finite normal-form game $\mathcal{G}$ has a $U_r$-equilibrium.
\end{corollary}

\section{Discussion through Examples}

To provide a better intuition, we  give  examples of well-known $2\times 2$ games from the basic game theory literature, and compare the outcomes of standard notions of equilibria against some notions of equilibria that we defined via local responses.

For convenience, the assumed metric is Euclidean; hence, 
if the opponent plays a mixed strategy $(x,1-x)$, the player believes that the strategy is anywhere in the set \\$\left\{(y,1-y) :  \max\{0,x-r\} \leq y \leq \max\{1,x+r\} \right\}$.

\subsection{Trembling-Hand Game}
\begin{figure}[ht!]
\begin{center}
\begin{tabular}{ r|c|c| }
\multicolumn{1}{r}{}
 &  \multicolumn{1}{c}{Left}
 & \multicolumn{1}{c}{Right} \\
\cline{2-3}
Up & 1, 1 & 2, 0 \\
\cline{2-3}
Down & 0, 2 & 2, 2 \\
\cline{2-3}
\end{tabular}
\end{center}
\caption{Trembling-Hand Game where both $\langle \text{Up, Left} \rangle$  and $\langle \text{Down, Right}\rangle$ are $PN$, yet only $\langle \text{Up, Left}\rangle$ is $T$.}
\label{trembling}
\end{figure} 
Call the example given in Figure \ref{trembling} Trembling-Hand Game for demonstration purposes. It seems that the pure strategy Nash equilibrium $PN=$\{(Up, Left), (Down, Right)\} while Trembling-hand perfect equilibrium $T$=\{(Up, Left)\} which matches with $W$.

Assume that $r_1 = r_2 \approx 0.14$. 
Now consider two mixed Nash strategy equilibria $\pi = \langle (1,0), (1, 0) \rangle$ and $\pi' = \langle(0, 1), (0, 1)\rangle$ where the former is also a $T$. See that for $\mathcal{B}_1(\pi)$, every strategy is dominated by $\pi_1$ in terms of $B_r$-best response and $W_r$-best response (analogous for the second agent). Moreover, the regret increases as  agent 1 diverges from $(1, 0)$. Hence $\pi \in B_r \cap W_r$.\footnote{Indeed, for such a  value of $r$, the strategy of the opponent varies only by 0.1. And the worst case is defined by the case that the opponent plays $(1, 0)$.} In the case of $\pi'$, it is a $B_r$-response since payoffs are already 2 for both agents. On the other hand, $\pi' \not\in W_r$ since worst case keeps improving for any agent who keeps  deviating.  Therefore, the regret also gets minimized since the best case value is fixed at  2.

\subsection{Matching Pennies} 

In the game of \emph{Matching Pennies}, $PN=\emptyset$ whereas The set of mixed strategy Nash equilibria is a singleton i.e.,  $MN=\{\pi\}$ where $\langle (0.5, 0.5), (0.5, 0.5)\rangle$.
\begin{figure}[h!]\label{matching}
\begin{center}
\begin{tabular}{ r|c|c| }
\multicolumn{1}{r}{}
 &  \multicolumn{1}{c}{Heads}
 & \multicolumn{1}{c}{Tails} \\
\cline{2-3}
Heads & 1, -1 & -1, 1 \\
\cline{2-3}
Tails & -1, 1 & 1, -1 \\
\cline{2-3}
\end{tabular}
\end{center}
\caption{Matching Pennies}
\end{figure}
 If $r \geq 0.5$ then all profiles are $B_r$-equilibrium. In the case of $U_r$ any strategy in the set $(0.5-r,0.5+r)$ (coupled symmetrically with the other player) forms an equilibrium.  This means an imprecise randomization can also be an equilibrium, provided that players have some level of uncertainty over the exact randomization of the other player.

\subsection{Stag  Hunt}

\begin{figure}[ht!]
\begin{center}
\begin{tabular}{ r|c|c| }
\multicolumn{1}{r}{}
 &  \multicolumn{1}{c}{Stag}
 & \multicolumn{1}{c}{Hare} \\
\cline{2-3}
Stag & 5, 5 & -1, 3 \\
\cline{2-3}
Hare & 3, -1 & 1, 1 \\
\cline{2-3}
\end{tabular}
\end{center}
\caption{Stag-Hunt Game}
\label{stag_hunt}
\end{figure}

We continue with the rather well-known coordination game 
Stag-Hunt illustrated in Figure~\ref{stag_hunt}.

Pure strategy-Nash equilibria are $S=\langle \text{Stag, Stag}\rangle$ and $H=\langle \text{Hare, Hare} \rangle$. There is also a mixed equilibrium that we will not consider. Playing $S$ is more socially desirable, but less stable according to several solution concepts based on uncertainty and risk aversion (see the work by \cite{carlsson1993global} for an overview and discussion). However these notions are not sensitive to possible differences between players' perceptions and do not  quantify the  instability of $S$. In contrast, such quantification  is very natural when we consider distance-based equilibria, by the size $r$ of the respective belief set. For example, for pessimistic players (who attribute excessive probability to the opponent playing ``Hare"), $S$ is a $W_r$-equilibrium for any $r$, but $H$ is only a $W_r$-equilibrium if $r\leq \frac13$.

\section{A Bliss of Ignorance}

Once a new notion of equilibrium is introduced, it is natural to ask whether it leads to any efficient outcome in the game. In this section, we explore this question and provide with an exemplary class of normal-form games that this is indeed the case.  In doing so, we employ the notion of \emph{Price of Anarchy} which is central as a measure of how much a system becomes inefficient due to selfish behaviour \cite{nisan2007algorithmic} (recall the preliminaries section for formal definition). The class of games we introduce is a consensus game \cite{balcan2009improved} which is asymmetric in payoffs.

\begin{definition}[Consensus Game]
A normal-form game $\mathcal{G}=(N, A, u)$ where $A_i \geq 2$ for every $i \in N$ is called a \emph{consensus game}  if  $u_i(a')=c'$ and $u(a)=c$ for every $a\in A \setminus \{a'\}$ with $c'>c$.
\end{definition}

Intuitively, the given consensus game is a coordination game in which only a single pure strategy profile has a higher payoff  i.e., $c'$ for every player compared to all the other pure strategy profiles which has $c$. Such game model (group) scenarios in which every member player has to agree unilaterally a decision to be taken (e.g., World Trade Org.).

One can observe that any given consensus game has at least two pure strategy equilibria $a'$ and $a$ such that only one of them has a more desirable outcome i.e., $SW(a')=n\cdot c'$,  and $SW(a)=n\cdot c$. The following result shows that  an ignorance factor $r > 0$, eliminates the undesirable equilibrium.

\begin{proposition}
Every consensus game has a unique $D_r$ equilibrium where $r_i>0$ for all $i \in N$. Moreover, PoA is 1. 
\end{proposition}

\begin{proof}
Observe that profile $a'$ has the best possible payoff for every single agent, hence it is an equilibrium. Moreover, due to the linearity of utilities  (i.e., $\pi_i(a') c' + (1-\pi_i(a')) c >  \pi'_i(a') c' + 1- \pi'_i(a') c$ whenever $\pi_i(a') > \pi'_i(a')$, pure strategy $\pi_i(a') = 1$ locally dominates every other strategy (i.e., $\alpha = 1$) for any $\mathcal{B}_i(\pi'',r)$ with $r>0$, hence it is unique (since, by Definition \ref{def:LD} there cannot be two distinct best response which can locally dominate each other). As it is the best possible outcome, PoA becomes 1. 
\end{proof}
Exploring such scenarios and extending them to more general class of games is left as future work. Yet still to develop a general understanding, it is important to look at PoA from the lenses of distance-based uncertainty.  In this regard, we deliver our final technical result.  In particular, we provide a bound (in terms of $r$) on the gain/loss of social welfare in an equilibrium modulo strict uncertainty.

The smoothness framework provides a convenient tool to bound the PoA in games~\cite{roughgarden2009intrinsic}:
If there are $\lambda,\mu>0$ s.t. for any two pure profiles $a,a'$ we have 
$$\sum_{i\in N}u_i(a'_{i}, a_{-i}) \geq \lambda \sum_{i\in N}u_i(a') - \mu\sum_{i\in N}u_i(a),$$
then for any pure/mixed/correlated/coarse-correlated equilibrium $\pi^*$ and any profile $\vec a$:   
$$\frac{SW(\pi^*)}{SW(a')} \geq \frac{\lambda}{1+\mu}.$$

The proof is trivial for pure equilibria. 
Now, the question we ask is ``can we extend this result to $\star_r$ equilibria (perhaps with a relaxed bound)"?

For a game $\mathcal{G}$, let 
\begin{multline}
\delta_G(r)=\max\{\max\left\{\frac{u_i(a_i,\pi_{-i})}{u_i(a_i,\pi'_{-i})},
\frac{u_i(a_i,\pi'_{-i})}{u_i(a_i,\pi_{-i})}\right\} \\: i\in N, a_i\in A_i, \pi'_{-i}\in B_i(\pi_{-i},r)\},\end{multline} i.e., the maximal utility ratio of an agent within a sphere of radius $r$.

\begin{theorem}\label{poa}
If there are $\lambda,\mu>0$ s.t. for any two pure profiles $a,a'$ we have 
$$\sum_{i\in N}u_i(a'_i, a_{-i}) \geq \lambda \sum_{i\in N}u_i(a') - \mu\sum_{i\in N}u_i(a),$$
then for any $\star_r$-pure equilibrium $a^*$ (for any $\star\in \{U,D,W,B\}$) and any profile $a'$:   
$$\frac{SW(a^*)}{SW(a')} \geq \frac{\lambda}{\delta_G(r)^2+\mu}.$$
\end{theorem}
Due to space limitations, we move the proof to the appendix.

 \section{Conclusion and Future Avenues}
 
 We have introduced a distribution-free agent model based on strict uncertainty, and studied consequent equilibria notions under different best-response behaviours. In the context of normal-form games, we explored the links between the notions we defined and a handful of existing well-known solution concepts which model mistakes and imprecision such as Trembling-hand perfect equilibrium (variants) and Robust equilibrium. For instance, it is shown that our notion is naturally generalizes Robust equilibrium.   It seems that strict equilibrium notion $D_r$ does not exist in general while all other entailed distance-based notions exist. Complementing those existence results with complexity results is an interesting line of future work.

 We looked for a possible scenario in which such solution concepts could potentially be useful, and introduced a coordination game in which ignorance was indeed helpful for the players to avoid a worst-outcome. Investigating more general game classes that distance-based  uncertainty solutions give rise to nice outcome guarantees deserves a further study on its own, and is our high priority for future research. As a more general outlook, we showed how to bound the loss of social welfare in any equilibrium (PoA) as uncertainty grows in terms of ignorance factor $r$. It would be nice to obtain finer bounds for games with different local-best responses.

Moreover, studying these notions on certain classes of games e.g., repeated games, as well as extending to extensive form games in general is our future research agenda.

\bibliographystyle{aaai}
\bibliography{aaai}

\end{document}